\documentclass[12pt]{article}
\usepackage[a4paper]{geometry}
\usepackage[utf8]{inputenc}
\usepackage{graphicx}
\usepackage{dsfont}
\usepackage{fancyhdr}
\usepackage{booktabs}
\usepackage[colorlinks=true, allcolors=blue,pagebackref=true]{hyperref}
\usepackage{xurl}
\usepackage{xcolor}
\usepackage{colortbl}
\usepackage{subfig}
\usepackage{mathtools}
\usepackage[T1]{fontenc}
\usepackage{comment}
\usepackage{natbib}
\usepackage{eurosym}
\usepackage{babel}
\usepackage{amsmath}
\allowdisplaybreaks
\usepackage{xcolor}
\usepackage{enumitem}
\usepackage{float}
\usepackage{here}
\usepackage{wrapfig}
\usepackage{lscape}
\usepackage{bbold}
\usepackage{stmaryrd}
\usepackage[export]{adjustbox}
\usepackage{tabularx}
\usepackage{amssymb}
\usepackage{mathtools}
\usepackage{multirow,makecell, framed}
\usepackage[french,onelanguage, boxruled, vlined]{algorithm2e}
\usepackage{amsthm, bm}
\usepackage{lscape}
\usepackage{ulem}
\usepackage{vmargin}

\newcommand{\pkg}[1]{{\fontseries{m}\fontseries{b}\selectfont #1}}

\newcommand{\dint}{\displaystyle\int}
\newcommand{\dsum}{\displaystyle\sum}

\newtheorem{lemma}{Lemma}
\newtheorem{proposition}{Proposition}
\newtheorem{theorem}{Theorem}

\theoremstyle{remark}

\usepackage{authblk}

\renewcommand*{\backref}[1]{}
\renewcommand*{\backrefalt}[4]{}

\setmarginsrb{2cm}{1cm}{2cm}{1cm}{0.5cm}{0.5cm}{0.5cm}{0.5cm}

\title{A spatial scan statistical for categorical, functional data}
\author[,1]{Camille Frévent\thanks{Corresponding author: \texttt{camille.frevent@univ-lille.fr}}}
\affil[1]{Univ. Lille, CHU Lille, ULR 2694 - METRICS: Évaluation des technologies de santé et des pratiques médicales, F-59000 Lille, France}
\author[1]{Moustapha Sarr}
\author[2]{Sophie Dabo-Niang}
\affil[2]{Laboratoire Paul Painvelé UMR CNRS 8524, Inria-Datavers, University of Lille} 

\date{}

\begin{document}

\maketitle

\hrule
\section*{Abstract}
We have developed and tested a spatial scan statistic for categorical, functional data (CFSS) - a data structure within which current approaches cannot identify spatial clusters.  
Our methodology combines an encoding scheme for categorical, functional observations with a nonparametric scan statistic.  
In a simulation study with three distinct scenarios, the CFSS accurately recovered the simulated spatial clusters and gave very low false positive rates, high true positive rates, and high positive predictive values.  
We have also used the CFSS to identify and characterize spatial clusters in French air pollution data from the winter of 2024. \\[0.2cm]
\textbf{Keywords:} Categorical data, Cluster detection, Functional data, Spatial scan statistics \vspace{0.5cm}
\hrule

\section{Introduction}

Spatial cluster detection has long held a central position in spatial statistics; the primary objective is to develop statistical methodologies and computational procedures capable of identifying geographically localized regions in which the observed outcomes or characteristics differ significantly from those prevailing in the surrounding or overall study area. Cluster detection methods have a pivotal role in many domains (including epidemiology, environmental sciences, public health, and socio-economic research) in which the identification and characterization of spatial heterogeneity are essential for elucidating the underlying mechanisms and for making evidence-based decisions.

A spatial scan statistic is a widely used, methodologically robust statistical framework for the detection of spatial clusters. Using windows of varying locations and scales, the studied region is scanned for statistically significant clusters. There is no need for an \textit{a priori} specification of the clusters' position, spatial extent, or geometric configuration. This inherently data-driven approach makes spatial scan statistics particularly well suited to exploratory spatial analysis and the discovery of spatial patterns.

\cite{naus1965distribution} developed scan statistics first in a one-dimensional context and then in two dimensions \citep{naus1965clustering}. The technique was then applied to spatial data by \cite{kulldorff1995spatial} (for Bernoulli models) and \cite{kulldorff1997spatial} (for Poisson models). These foundational studies used a combination of likelihood ratio test statistics and Monte Carlo hypothesis testing to evaluate the statistical significance of identified clusters; this approach enables the detection of clusters of diverse sizes and shapes, while accounting for multiple testing.
Building on this early work, spatial scan statistics have been generalized to a broad range of data types and probabilistic models for continuous observations, such as Gaussian models \citep{kulldorff2009scan}, exponential models \citep{huang2007spatial}, and Weibull models \citep{bhatt2014spatial}. Many researchers have gone beyond the univariate case and have developed extensions for more complex data structures. Notably, \cite{cucala2017multivariate} developed a spatial scan statistic grounded in a multivariate normal model that explicitly incorporates correlations among variables, and \cite{cucala2019multivariate} introduced a multivariate nonparametric method based on a Wilcoxon–Mann–Whitney-type test statistic.

Alongside these methodological advances, progress in sensing technologies, electronic health records, and data storage infrastructures has enabled the acquisition of increasingly rich datasets in which information is captured continuously in time and space. This shift has fostered the development of functional data analysis \citep{ramsaylivre}, a statistical paradigm that models observations as realizations of random functions instead of finite-dimensional vectors. Consequently, many classical statistical techniques have been revisited and adapted to a functional context, including spatial scan statistics for functional data \citep{frevent2021detecting,frevent2023investigating}.

In many practical settings, however, the observed data are not real-valued. In particular, categorical variables arise frequently in the health sciences, when (for example) disease stages, treatment regimens, or clinical endpoints are recorded. Several methods for spatial cluster detection have been developed in these contexts. Most notably, \cite{jung2007spatial} and \cite{jung2010spatial} introduced spatial scan statistics based respectively on ordinal and multinomial modeling frameworks. These procedures are effective for detecting spatial heterogeneity in purely categorical spatial datasets.

In contrast, observations of patient trajectories and longitudinal health monitoring are often both categorical and functional in nature because the repeated recording of patients' health states yields time indexed, categorical-valued functions. The statistical analysis of this type of data involves challenges in both functional data analysis and categorical data modeling; one must simultaneously accommodate time dependence, the inherently functional nature of the observations, and the discrete, finite outcome space. To the best of our knowledge, none of the spatial cluster detection methods developed to date accounts rigorously and concurrently for all these features. Today's functional spatial scan statistics \citep{frevent2021detecting,frevent2023investigating} are designed for real-valued functional data, whereas categorical spatial scan methods \citep{jung2007spatial,jung2010spatial} are restricted to scalar, non-functional observations.

We addressed this gap by developing a novel spatial scan statistic that is specifically tailored to categorical, functional data and is capable of identifying spatial clusters in which the change over time in categorical states diverges significantly from that observed in the surrounding regions. Our approach simultaneously addresses three main methodological challenges. Firstly, it explicitly accommodates the functional nature of the observations and thereby captures time dependencies in individual trajectories. Secondly, it models the categorical outcome space within a statistically coherent, probabilistic framework. Thirdly, it incorporates spatial structure for the detection of geographically localized clusters without prior knowledge of their spatial location or extent. By integrating functional data representations, probabilistic modeling of categorical trajectories, and rigorous Monte Carlo-based inferential procedures, the new framework generalizes classical spatial scan statistics to a broader, more relevant class of problems.

The methodology has significant practical implications - especially in the health sector, where capturing spatial heterogeneity in patient trajectories or care pathways is crucial for designing targeted interventions and enhancing healthcare quality. The ability to pinpoint areas with unusual categorical trajectory patterns enables public health authorities and clinicians to identify emerging problems, allocate resources more effectively, and adapt interventions to local specificities.
The remainder of the present manuscript is organized as follows. The methodological framework that underlies the new spatial scan statistic for categorical, functional data (CFSS) is described in Section~\ref{sec:method}. The scan statistic's performance in simulation studies and its application to an empirical dataset are described in Section~\ref{sec:simu}. Lastly, concluding remarks, possible extensions, and directions for future research are discussed in Section~\ref{sec:discussion}.

\section{Methodology} \label{sec:method}

\subsection{General principle}

Let $\{X(t), X(t) \in E=\{e_1, \dots, e_J\}, t \in [0,T] \}$ be a categorical, continuous-time stochastic process where $E$ is a set of $J \ge 2$ states and $[0,T]$ is an interval of $\mathbb{R}$. \\
Let $s_1, \dots, s_K$ be $K$ non-overlapping locations of an observation domain $S \subset \mathbb{R}^2$ and $X_{1,1}\mathrel{,} \dots\mathrel{,} X_{1,n_1}\mathrel{,} \dots\mathrel{,} X_{K,1}\mathrel{,} \dots\mathrel{,} X_{K,n_K}$ be
the observations of $X$ in $s_1, \dots, s_K$, where $X_{k,i}$ denotes the $i\text{th}$ individual of spatial unit $s_k$. Hereafter, all observations are considered to be independent - a classical assumption in scan statistics. \\

The goal of spatial scan statistics is to identify spatial clusters and assess their statistical significance. This involves a null hypothesis $\mathcal{H}_0$ (the absence of a cluster) and a composite alternative hypothesis $\mathcal{H}_1$ (the presence of at least one cluster $w \subset S$
presenting abnormal values of $X$). \\

A spatial cluster is defined as a geographical area in which the profile of the individuals therein differs from that of the rest of the population.
In the case of categorical, functional data, several types of cluster can be considered: (i) the sequence of states in individuals in the cluster differs from that in individuals outside the cluster, (ii) individuals in the cluster change state more (or less) rapidly than individuals outside the cluster, and (iii) individuals in the cluster spend more (or less) time in certain states than individuals outside the cluster. \\

\begin{figure}[h]
\centering
\includegraphics[width=0.8\linewidth]{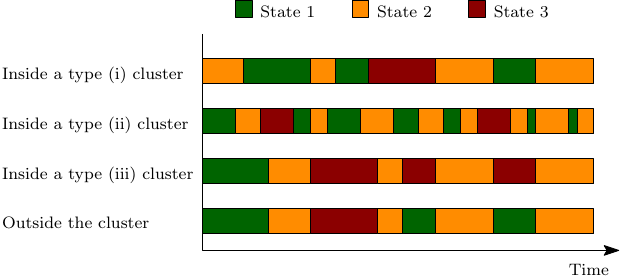}
\caption{An illustration of three types of clusters in categorical, functional data. Inside a type (i) cluster, the individuals' sequence of states differs from that outside the cluster; here, the first state differs, and the individuals move directly from state 1 to state 3 without passing through state 2. In the type (ii) cluster, individuals change state more (or less) rapidly than individuals outside the cluster. In the type (iii) cluster, the individuals spend more (or less) time in certain states than individuals outside the cluster.}
\label{fig:typesclusters}
\end{figure}

Following \cite{cressie1977some}, a scan statistic is defined by the maximum of a concentration index over a set of potential clusters $\mathcal{W}$. In the following and without loss of generality, we focus on variable-size, circular clusters (as introduced by \cite{kulldorff1997spatial}). The set of potential clusters $\mathcal{W}$ is the set of discs centered on a location and passing through another one:
$$ \mathcal{W} = \left\{ w_{k,k'} / 1 \le |w_{k,k'}| \le \frac{n}{2}, 1 \le k,k' \le K \right\}, $$
where $w_{k,k'}$ is the disc centered on $s_k$ and that passes through $s_{k'}$, $|w_{k,k'}|$ is the number of individuals in $w_{k,k'}$, and $n = \sum_{k=1}^K n_k$ is the total number of observed individuals. Thus, a given cluster never contains more than 50\% of the studied individuals, as recommended by \cite{kulldorff1995spatial}.
It should be noted that elliptical clusters \citep{kulldorff2006elliptic} and irregularly-shaped clusters \citep{tango2005flexibly,yin2018hybrid} have been suggested in the literature.

\subsection{Dimension reduction} \label{sec:dimreduction}

The major problem with categorical, functional data is the paucity of statistical methods for process them. The transformation of categorical, functional data into real values enables the latter to be processed more easily for spatial cluster detection. To this end, we use the multiple correspondence analysis for functional data (introduced by \cite{preda2021categorical}) to transform the $X_{k,i}$ into real vectors $\bm{Z}_{k,i} =(Z_{k,i}^{(1)}, \dots, Z_{k,i}^{(M)})^\top$, while retaining as much information as possible.

Let denote 
$\mathds{1}_{e_j}(t) = \left\{ \begin{array}{cc}
1 & \text{ if } X(t) = e_j \\
0 & \text{ if } X(t) \neq e_j \\
\end{array} \right.$.

We assume that ($\mathcal{A}_1$) $X$ is continuous in probability: $ \underset{h \rightarrow 0}{\lim} \ \mathbb{P}(X(t+h) \neq X(t)) = 0 $ and that ($\mathcal{A}_2$) $ \forall t \in [0,T], \forall e_j \in E, \mathbb{P}(X(t) = e_j) \neq 0 $. \\

The functional multiple correspondence analysis finds new orthogonal variables $Z^{(1)}, \dots, Z^{(M)}$ that maximize
\begin{equation}
\dint_0^T \dfrac{\mathbb{V}(E_t(Z))}{\mathbb{V}(Z)} \ \text{d}t
\label{eq:maximize}
\end{equation}

where $E_t(Z) = \dsum_{j=1}^J \mathbb{E}(Z|X(t) = e_j) \mathds{1}_{e_j}(t)$ is the conditional expectation operator on $L^2$, the space of real random variables with finite second moment associated with $X(t)$, $\mathbb{V}(E_t(Z))$ is the variance of $E_t(Z)$.

\begin{proposition} \phantom{newline} \\
$E_t$ is a self-adjoint, idempotent operator.
\label{prop:propEt}
\end{proposition}

Maximizing (\ref{eq:maximize}) is equivalent to solving
\begin{equation}
\dint_0^T E_t(Z) \text{d}t = \lambda Z.
\label{ref:eq1}
\end{equation}
Then $Z^{(1)}, \dots, Z^{(M)}$ are orthogonal eigenvectors of $\dint_0^T E_t(\cdot) \ \text{d}t$ associated with the eigenvalues $\lambda^{(1)}, \dots, \lambda^{(M)}$ in descending order, excluding the trivial solution $Z^{(0)}=1$ associated with the trivial eigenvalue $\lambda^{(0)}=T$. \\ 

To find the non-trivial eigenvectors and eigenvalues, \cite{preda2021categorical} defined $\zeta(t) = \dfrac{1}{\lambda} E_t(Z) \in H([0,T])$ the Hilbert space of real processes with finite total variance $ \dint_0^T \mathbb{V}(\zeta(t)) \ \text{d}t $ and equipped with the scalar product $\langle \zeta, \zeta' \rangle_H = \dint_0^T \mathbb{E}[\zeta(t) \zeta'(t)] \ \text{d}t $.

When considering the conditional expectation with respect to $X(t')$, (\ref{ref:eq1}) can be rewritten as 
\begin{equation}
\dint_0^T K(t',t) \zeta(t) \ \text{d}t = \lambda \zeta(t')
\label{ref:eq2}
\end{equation}
with $K(t',t) = E_{t'}E_t$.

$\zeta(t)^{(0)} = \dfrac{1}{T}, t \in [0,T]$ is a trivial eigenfunction of $K$ associated with the trivial eigenvalue $\lambda^{(0)} = T$, and $\zeta(t)^{(1)}, \dots,  \zeta(t)^{(M)}$ are the non-trivial orthogonal (for $\langle \cdot, \cdot \rangle_H$) eigenfunctions associated with eigenvalues $\lambda^{(1)}, \dots, \lambda^{(M)}$, in descending order. \\

To ensure the uniqueness of $\zeta(t)^{(1)}, \dots,  \zeta(t)^{(M)}$, \cite{saporta1981methodes} and \cite{preda2021categorical} added the following constraint:
\begin{equation}
\dint_0^T \mathbb{E}\left[ \left(\zeta(t)^{(m)}\right)^2 \right] \ \text{d}t = 1.
\label{eq:constraint1}
\end{equation}

The eigenfunctions' properties are given in the following propositions.
\begin{proposition}[\cite{saporta1981methodes} p107]
\phantom{newline} \\
$\dint_0^T \mathbb{E}\left[ \zeta(t)^{(1)} \right] \ \text{d}t = \dots = \dint_0^T \mathbb{E}\left[ \zeta(t)^{(M)} \right] \ \text{d}t = 0$.
\label{prop:meanzeta}
\end{proposition}

\begin{proposition}[\cite{preda2021categorical}]
\phantom{newline} \\
$\zeta(t)^{(m)}$ can be written as
$$ \zeta(t)^{(m)} = \dsum_{j=1}^J a_j(t)^{(m)} \mathds{1}_{e_j}(t) $$
where $\{a_j(\cdot)^{(m)}, j = 1, \dots, J \}$ are deterministic encoding functions on $[0,T]$. Thus, finding $\zeta(\cdot)^{(m)}$ that satisfies (\ref{ref:eq2}) is equivalent to solving the more classical problem
\begin{equation}
\dint_0^T \dsum_{j=1}^J \mathbb{P}(X(t') = e_{j'}, X(t) = e_j) a_j(t)^{(m)} \ \text{d}t = \lambda a_{j'}(t')^{(m)} \mathbb{P}(X(t') = e_{j'}), 
\label{ref:eq3}
\end{equation}
and constraint (\ref{eq:constraint1}) becomes
$$ \dint_0^T \dsum_{j=1}^J \left[ a_j(t)^{(m)} \right]^2 \mathbb{P}(X(t) = e_j) \ \text{d}t = 1. $$
\label{prop:prop2}
\end{proposition}

\begin{proposition}[\cite{saporta1981methodes}]
\phantom{newline} \\
From $\zeta(t)^{(0)}, \dots, \zeta(t)^{(M)}$, the solutions $Z^{(0)}, \dots, Z^{(M)}$ of (\ref{ref:eq1}) can be deduced by
$$ Z^{(m)} = \dint_0^T \zeta(t)^{(m)} \text{d}t. $$
\label{prop:prop3}
\end{proposition}

\begin{proposition}[\cite{saporta1981methodes}] 
\phantom{newline} \\
The non-trivial solutions of (\ref{ref:eq1}) $Z^{(1)}, \dots, Z^{(M)}$ are centered, uncorrelated random variables, such that $\mathbb{V}(Z^{(m)}) = \lambda^{(m)}, 1 \le m \le M$.
\label{prop:prop4}
\end{proposition}

We provide detailed versions of the proofs from \cite{saporta1981methodes} and \cite{preda2021categorical} in Supplementary Materials.

\subsection*{Multiple correspondence analysis in practice}

The practical computation of $\zeta_{1,1}^{(1)}$, $\dots$, $\zeta_{1,1}^{(M)}$, $\dots$, $\zeta_{K,n_K}^{(1)}$, $\dots$,  $\zeta_{K,n_K}^{(M)}$, and thus $Z_{1,1}^{(1)}$, $\dots$, $Z_{1,1}^{(M)}$, $\dots$, $Z_{K,n_K}^{(1)}$, $\dots$,  $Z_{K,n_K}^{(M)}$ from the observed data $X_{1,1}, \dots, X_{K,n_K}$ is usually obtained by approximating the encoding functions $\{a_{j, k, i}(\cdot)^{(m)}, 1 \le m \le M, 1 \le j \le J, 1 \le k \le K, 1 \le i \le n_k \}$ into an arbitrary finite basis of functions. 
The theoretical and computational details are presented in \cite{saporta1981methodes} and \cite{preda2021categorical}, and 
the procedure is implemented in the \textsf{R} package \pkg{cfda}. \\

In the following we select $M$ such that at least 90\% of the variance is explained.

\subsection{The CFSS}

For now, the data $X_{1,1}, \dots, X_{1,n_1}, \dots, X_{K,1}, \dots, X_{K,n_K}$ have been encoded into $M$-dimensional real (non-functional) data $\bm{Z}_{1,1}, \dots, \bm{Z}_{1,n_1}, \dots, \bm{Z}_{K,1}, \dots, \bm{Z}_{K,n_K}$. 

We can therefore adapt the nonparametric spatial scan statistic for multivariate data (developed by \cite{cucala2019multivariate} and implemented in the \textsf{R} package \pkg{HDSpatialScan} \citep{frevent2022r,HDSpatialScan}) to the present context.

We consider the following test hypotheses:
\begin{center}
$\mathcal{H}_0$: $\bm{Z}_{1,1}, \dots, \bm{Z}_{1,n_1}, \dots, \bm{Z}_{K,1}, \dots, \bm{Z}_{K,n_K}$ are identically distributed,
whatever their associated locations $s_1, \dots, s_K$ \\
vs. \\
$\mathcal{H}_1$: there exists $w \in \mathcal{W}$ such that the distribution of $\bm{Z}$ in $w$ differs from that of $\bm{Z}$ outside $w$. \\
\end{center}

Let
$$ R_{k,i} = \dfrac{1}{n} \dsum_{k'=1}^K \dsum_{i'=1}^{n_{k'}} \text{sgn}\left[A_{\bm{Z}} (\bm{Z}_{k,i} - \bm{Z}_{k',i'}) \right] $$ be the multivariate ranks where 
$ \text{sgn}(\cdot) $ is the multivariate sign function
$$ \begin{array}{rccl}
\text{sgn}: & \mathbb{R}^M & \rightarrow & \mathbb{R}^M \\
& \bm{z} & \mapsto & ||\bm{z}||_2^{-1} \bm{z} \ \mathds{1}_{\bm{z} \neq \bm{0}_M}
\end{array} $$
with $\bm{0}_M$ the $M$-dimensional vector composed only of 0, and $A_{\bm{Z}}$ is a data-based transformation matrix that makes the multivariate
ranks behave as though they were spherically distributed in the unit $M$-sphere:
$$ \dfrac{M}{n} \dsum_{k=1}^K \dsum_{i=1}^{n_k} R_{k,i} R_{k,i}^\top = \dfrac{1}{n} \dsum_{k=1}^K \dsum_{i=1}^{n_k} R_{k,i}^\top R_{k,i} I_M. $$

The test statistic
$$ W^{(w)} = \dfrac{Mn}{\dsum_{k=1}^K \dsum_{i=1}^{n_k} R_{k,i}^\top} \left( |w| || \bar{R}_w ||_2^2 + |w^\mathsf{c}| || \bar{R}_{w^\mathsf{c}} ||_2^2 \right), $$
where $\bar{R}_g = \dfrac{1}{|g|} \dsum_{\substack{k=1 \\ s_k \in g}}^K \dsum_{i=1}^{n_k} R_{k,i}, g \in \{w, w^\mathsf{c}\} $ allows one to test \citep{oja}:
\begin{center}
$\mathcal{H}_0: $ the distribution of $\bm{Z}$ is the same inside and outside $w$ \\
vs. \\
$\mathcal{H}_1: $ the distribution of $\bm{Z}$ in $w$ differs from that of $\bm{Z}$ outside $w$.
\end{center}

Then, following \cite{cucala2019multivariate}, we consider $ W^{(w)} $ as a concentration index for cluster detection and define the CFSS as
$$ \Lambda_\text{CFSS} = \underset{w \in \mathcal{W}}{\max} \ W^{(w)}. $$ 

The most likely cluster (MLC) is the potential cluster for which this maximum is obtained:
$$ \text{MLC}_\text{CFSS} = \underset{w \in \mathcal{W}}{\arg \max} \ W^{(w)}. $$

\subsection{Computing the statistical significance of the MLC}

Once the MLC has been identified, its statistical significance must be evaluated. The distribution of the scan statistic under $\mathcal{H}_0$ is intractable, due to the inherent dependency between statistics $W^{(w)}$ and $W^{(w')}$ when $w \cap w' \neq \emptyset$. \\

To estimate the p-value associated with the observed spatial scan statistic, we employ an approach based on random permutations and that is commonly referred to as ``random labelling''. 

Let $P$ denote the number of random permutations performed on the original dataset, and $\Lambda_\text{CFSS}^{(1)}, \dots, \Lambda_\text{CFSS}^{(P)}$ denote the spatial scan statistics obtained from these simulated datasets. According to \cite{dwass1957modified}, the estimated p-value for the scan statistic observed in the original data is calculated as follows:
$$ \widehat{p} = \dfrac{1 + \dsum_{p'=1}^P \mathds{1}_{\Lambda_\text{CFSS}^{(p')} \ge \Lambda_\text{CFSS}}}{1+P}. $$ 
 
Lastly, the MLC is considered statistically significant if its associated estimated p-value is less than a predefined type I error.

\section{Finite sample properties} \label{sec:simu}

\subsection{Simulation study}

We generated artificial datasets using the geographical coordinates of the 94 \textit{départements} (county-level administrative divisions) in mainland France. The location of each \textit{département} was defined by its centroid. A true cluster (denoted as $w$) was defined and simulated within each dataset. This cluster comprised the eight \textit{départements} located in the Paris metropolitan area (highlighted in red in the Supplementary Materials, Figure \ref{fig:map_simu}).

\subsubsection{Generation of the artificial datasets}

In each of the 94 spatial units, a total of $n_k = 10$ individual trajectories were simulated. Each trajectory corresponds to the realization of a three-state process $\{X_{k,i}(t) \in E = \{e_1, e_2, e_3\}, t \in [0,18] \} $ and was generated using a transition matrix $\bm{P}_k$ and a vector of intensities $\bm{\lambda}_k = (\lambda_{1,k}, \lambda_{2,k}, \lambda_{3,k})^\top$.

For each individual, the trajectory begins in $e_1$. The time spent in this state was generated according to an exponential distribution with rate $\frac{1}{\lambda_{1,k}}$ and the next state was chosen on the basis of the probabilities in the first row of $\bm{P}_k$. The same procedure was then repeated for all subsequent states (using the appropriate intensities from $\bm{\lambda}_k$ and the probabilities from $\bm{P}_k$) until the maximum time ($T=18$) was reached. \\

Three simulation scenarios were considered: (i) the transition matrices $\bm{P}_k$ do not differ inside and outside the spatial cluster, no state is absorbing, and the spatial cluster differs from the rest of the area in terms of intensity, (ii) the transition matrices $\bm{P}_k$ are the same inside and outside the spatial cluster, state $e_3$ is absorbing, and the spatial cluster differs from the rest of the area in terms of intensity, (iii) the transition matrices $\bm{P}_k$ are not the same inside and outside the spatial cluster, state $e_3$ is absorbing, and the intensities are the same inside and outside the spatial cluster.

For scenario (i), we define all $\bm{P}_k$ as
$$ \bm{P}_k = \bm{P} = \begin{pmatrix}
0 & 0.5 & 0.5 \\
0.5 & 0 & 0.5 \\
0.5 & 0.5 & 0
\end{pmatrix} 
$$
and $\bm{\lambda}_k = (0.2 + \alpha \mathds{1}_{s_k \in w},0.2 + \alpha \mathds{1}_{s_k \in w} ,0.2 + \alpha \mathds{1}_{s_k \in w})^\top$, $\alpha \in \{ 0, 0.15, 0.30, 0.45, 0.60, 0.75 \}$.

For scenario (ii), 
$$ \bm{P}_k = \bm{P} = \begin{pmatrix}
0 & 1 & 0 \\
0.5 & 0 & 0.5 \\
0.5 & 0.5 & 0
\end{pmatrix} 
$$
$\bm{\lambda}_k = (0.2 + \alpha \mathds{1}_{s_k \in w},0.2 + \alpha \mathds{1}_{s_k \in w} ,0)^\top$, $\alpha \in \{ 0, 0.05, 0.10, 0.15, 0.20, 0.25, 0.30 \}$.

For scenario (iii), $\bm{P}_k$ was defined as
$$ P_k = 
\begin{pmatrix}
0 & 1 & 0 \\
0.7 - \gamma \mathds{1}_{s_k \in w} & 0 & 0.3 + \gamma \mathds{1}_{s_k \in w} \\
0.5 & 0.5 & 0
\end{pmatrix} , \gamma \in \{ 0, 0.10, 0.20, 0.30, 0.40, 0.50  \} 
$$
and $\bm{\lambda}_k = \bm{\lambda} = (0.2,0.2,0)^\top$. 

It is noteworthy that $\alpha = 0$ and $\gamma = 0$ were also tested, in order to evaluate the maintenance of the nominal type I error. An example of the data for all scenarios is given in the Supplementary Materials (Figure \ref{fig:exemple_simu}).

\subsubsection{Evaluation of the method's performance}

For each scenario and each value of $\alpha$ (or $\gamma$), we simulated 1000 artificial datasets. The p-value associated with each MLC was estimated by generating 999 random permutations of the transformed data (for more details, see Section \ref{sec:dimreduction}). A type I error of 5\% was considered for the rejection of $\mathcal{H}_0$. Our CFSS's performance was evaluated with regard to four criteria: the power, the true positive rate, the false positive rate, and the positive predictive value.
The power was defined as the proportion of simulations leading to the rejection of $\mathcal{H}_0$, according to the type I error. For the simulated datasets yielding to the rejection of $\mathcal{H}_0$, the true positive rate is the mean proportion of correctly detected sites in $w$, the false positive rate is the mean proportion of sites in $w^\textsf{c}$ included in the detected cluster, and the positive predictive value corresponds to the proportion of individuals in $w$ within the detected cluster.

We compared the performances of our approach using 10, 20 and 30 basis functions when approximating the encoding functions $a_{j,k,i}(\cdot)^{(m)}$ (for more details, see Section \ref{sec:dimreduction}).

\subsubsection{Results of the simulation study}

\begin{figure}[h]
\centering
\includegraphics[width = \linewidth]{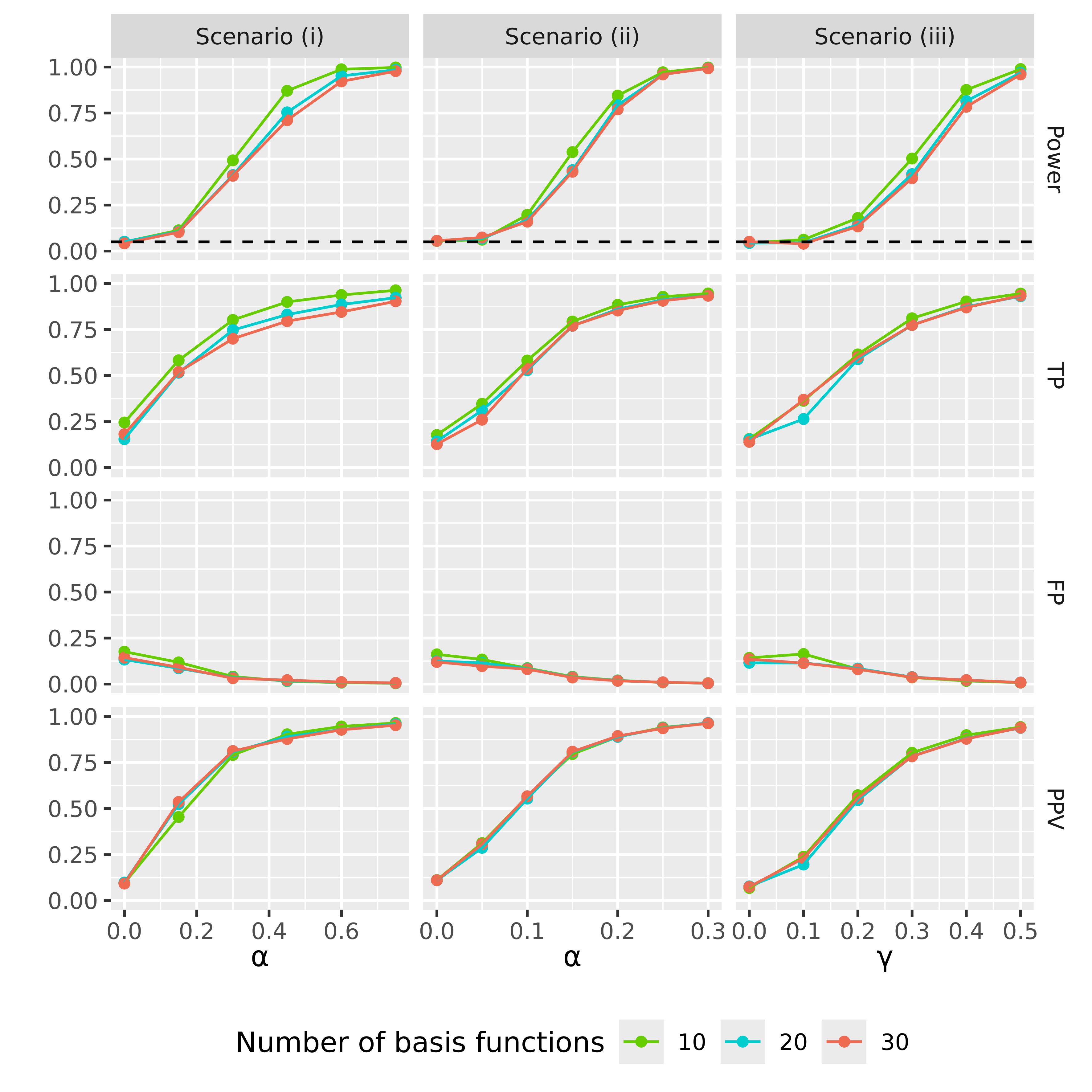}
\caption{Simulation of the CFSS's ability to detect a spatial cluster in categorical, functional data. For each scenario, the power curves, the true positive rates, the false positive rates, and the positive predictive values for detection of the spatial cluster as the MLC are shown. $\alpha$ and $\gamma$ are the parameters that control the cluster intensity.}
\label{fig:res_simu}
\end{figure}

For $\alpha = 0$ and $\gamma = 0$, the method maintained the correct type I error (0.05), regardless of the scenario and the number of basis functions (Figure \ref{fig:res_simu}). 
The number of basis functions had very little impact on performance, which is consistent with \cite{preda2021categorical}. 

The method demonstrated a satisfactory level of performance, as characterized by very low false positive rates, high true positive rates, and high positive predictive values. Furthermore, the approach successfully detected the simulated cluster in all scenarios considered.

\subsection{Application to real data} \label{sec:appli}

We considered the daily mean $\text{PM}_{10}$ concentrations in France between January $1^\text{st}$ and March $31^\text{st}$, 2024, using data from the Geod'air database (\url{https://www.geodair.fr/donnees/consultation}). The location of the 279 monitoring stations and the observed concentration curves are shown in Figure \ref{fig:presappli} (left panel). We then categorized the concentrations according to the Atmo index (see Table \ref{tab:atmo}, \url{https://www.atmo-nouvelleaquitaine.org/article/lindice-atmo-un-outil-precis-et-complet-de-la-qualite-de-lair}). To avoid very rare modalities, we grouped Atmo's ``bad'', ``very bad'' and ``extremely bad'' categories into a single ``bad'' category.

\begin{figure}[h!]
\centering
\begin{minipage}{0.59\linewidth}
\includegraphics[width=\linewidth]{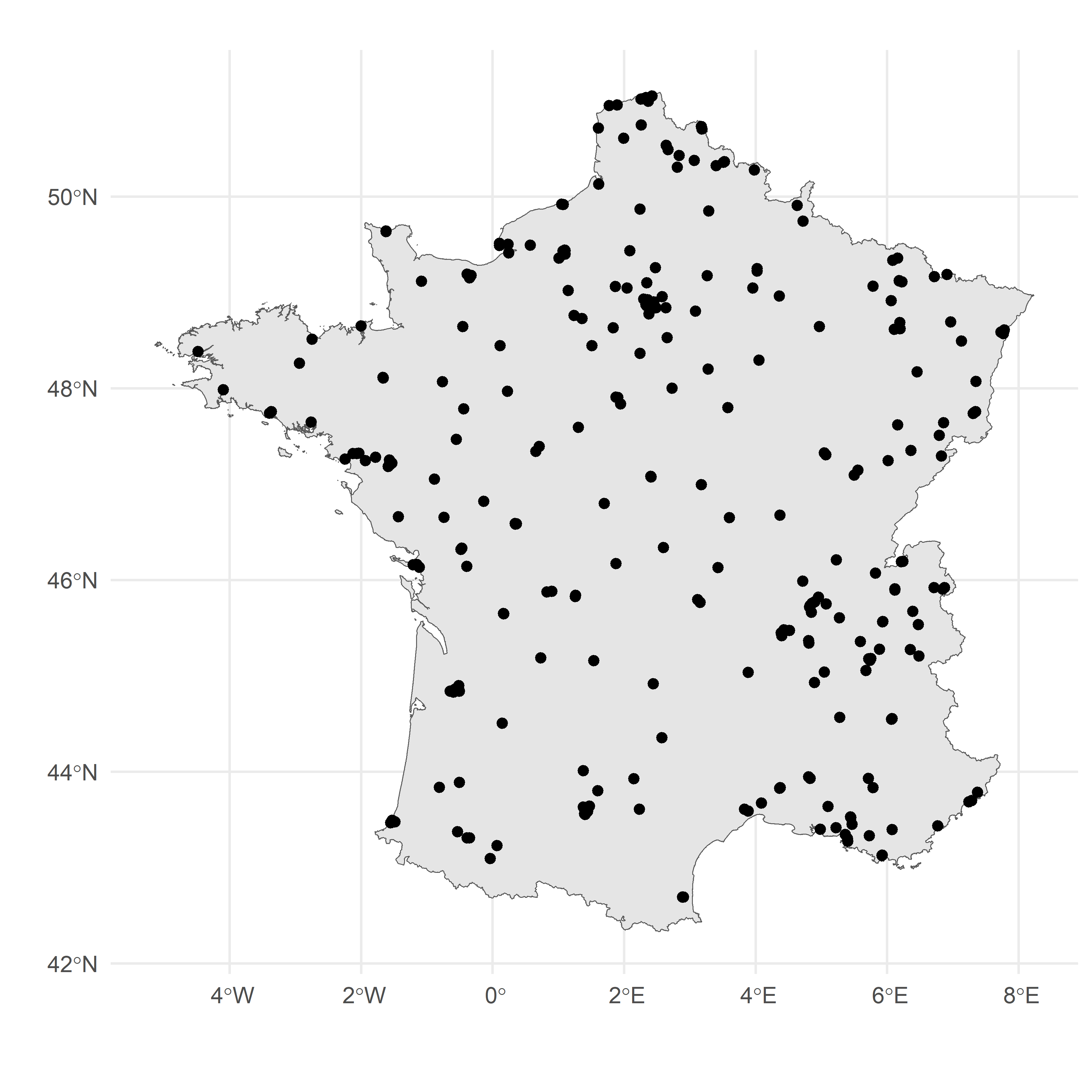} \\
\includegraphics[width=\linewidth]{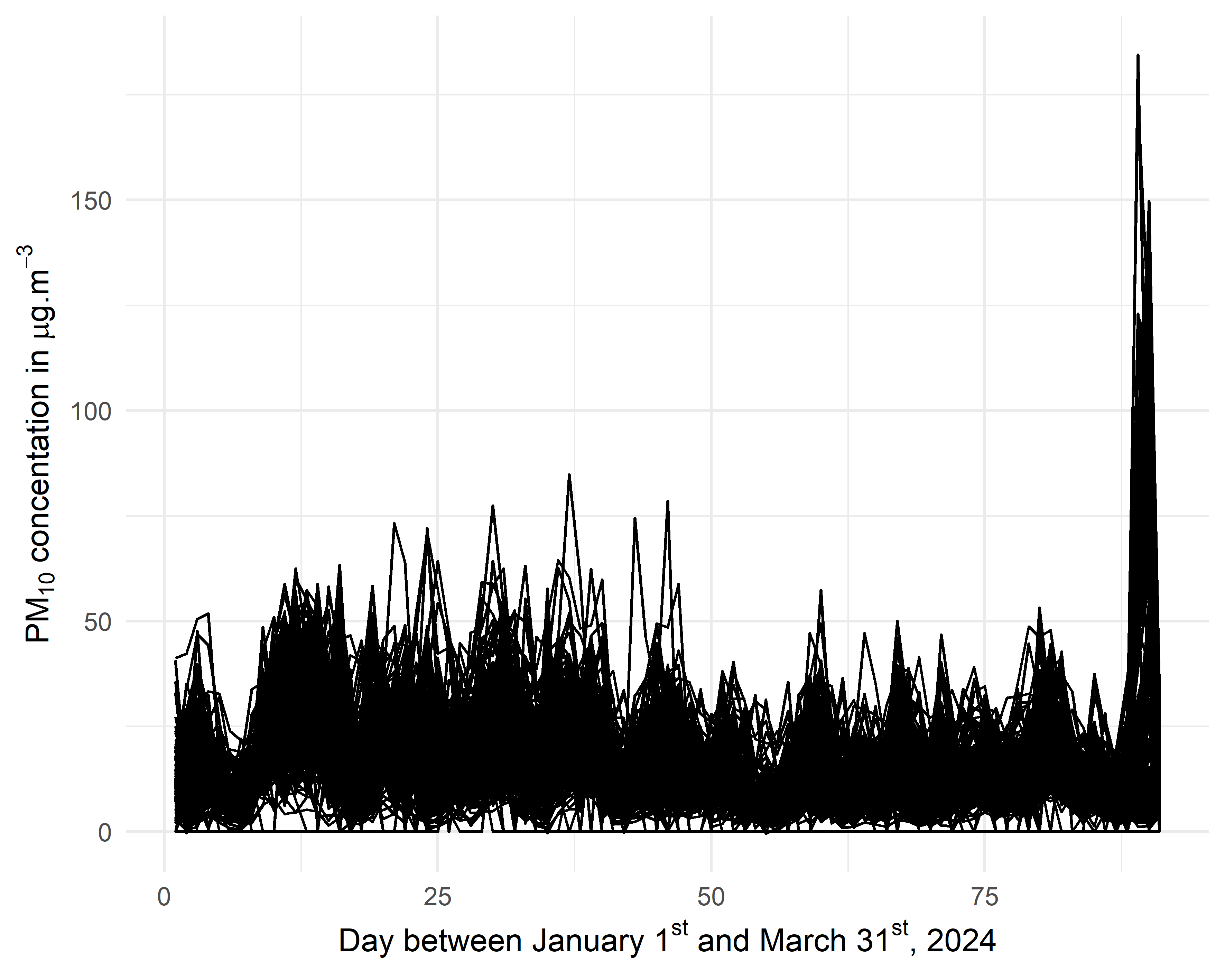}
\end{minipage} \hfill
\begin{minipage}{0.39\linewidth}
\includegraphics[width=\linewidth]{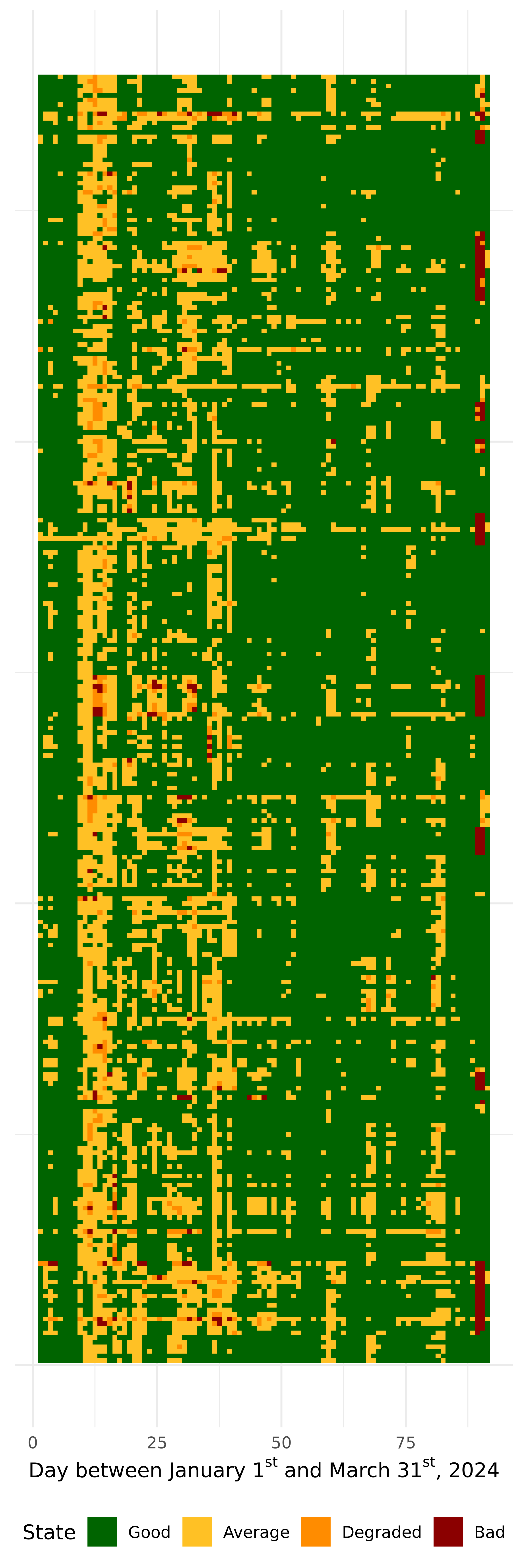}
\end{minipage}
\caption{Map of the monitoring stations and daily mean $\text{PM}_{10}$ concentrations curves between January $1^\text{st}$ and March $31^\text{st}$, 2024 (left panel) and the associated categorical trajectories (right panel).}
\label{fig:presappli}
\end{figure}

\begin{table}[h!]
\caption{The Atmo index, according to the range of daily mean $\text{PM}_{10}$ concentrations.}
\label{tab:atmo}
\centering
\begin{tabular}{cc}
\hline
\textbf{Range of daily mean $\text{PM}_{10}$ concentrations ($\mu g. m^{-3}$)} & \textbf{Atmo index} \\
$[0,20]$ & Good \\
$]20,40]$ & Average \\
$]40,50]$ & Degraded \\
$]50,100]$ & Bad \\
$]100,150]$ & Very bad \\
$>150$ & Extremely bad \\ \hline
\end{tabular}
\end{table}

We applied the CFSS by considering 10 basis functions when encoding the data. The considered set of potential clusters $\mathcal{W}$ was the set described in Section \ref{sec:method}.
The statistical significance of the MLC was evaluated in 999 Monte-Carlo permutations, and a spatial cluster was considered to be statistically significant when the associated p-value was below 0.05.

The detected MLC ($\widehat{p} = 0.001$) is located in southeastern France and comprises 71 stations (see Figure \ref{fig:mapMLC}). The associated trajectories (Figure \ref{fig:curveMLC}) show that relative to monitoring stations outside the MLC, the monitoring stations in the MLC recorded more transitions to ``average'', ``degraded'' and ``bad'' states and spent more time in these states.

\begin{figure}[h]
    \centering
    \includegraphics[width=0.6\linewidth]{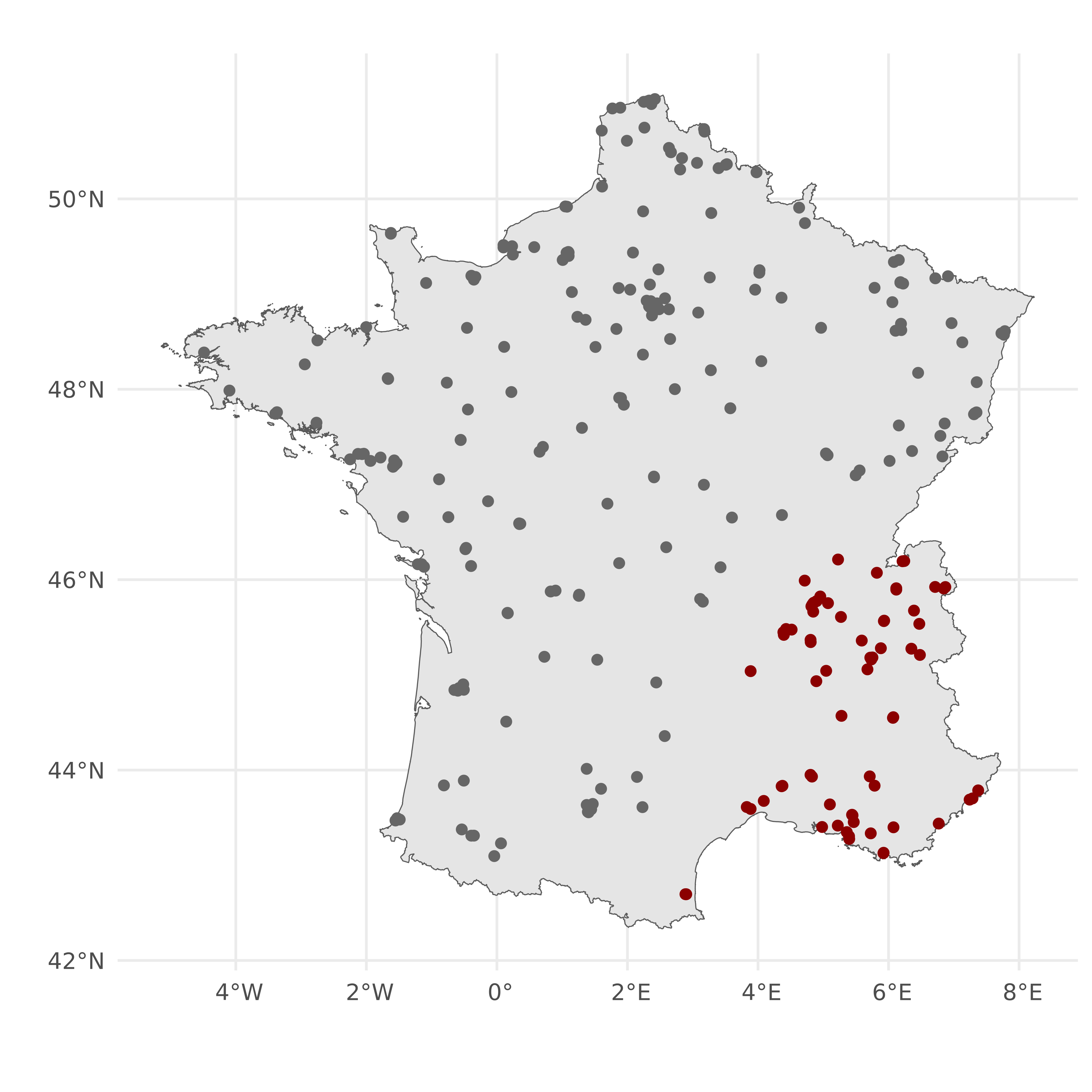}
    \caption{The MLC (in red) detected with the CFSS for the daily mean $\text{PM}_{10}$ concentrations (categorized using the Atmo index).}
    \label{fig:mapMLC}
\end{figure}

\begin{figure}[h]
    \centering
    \includegraphics[width=\linewidth]{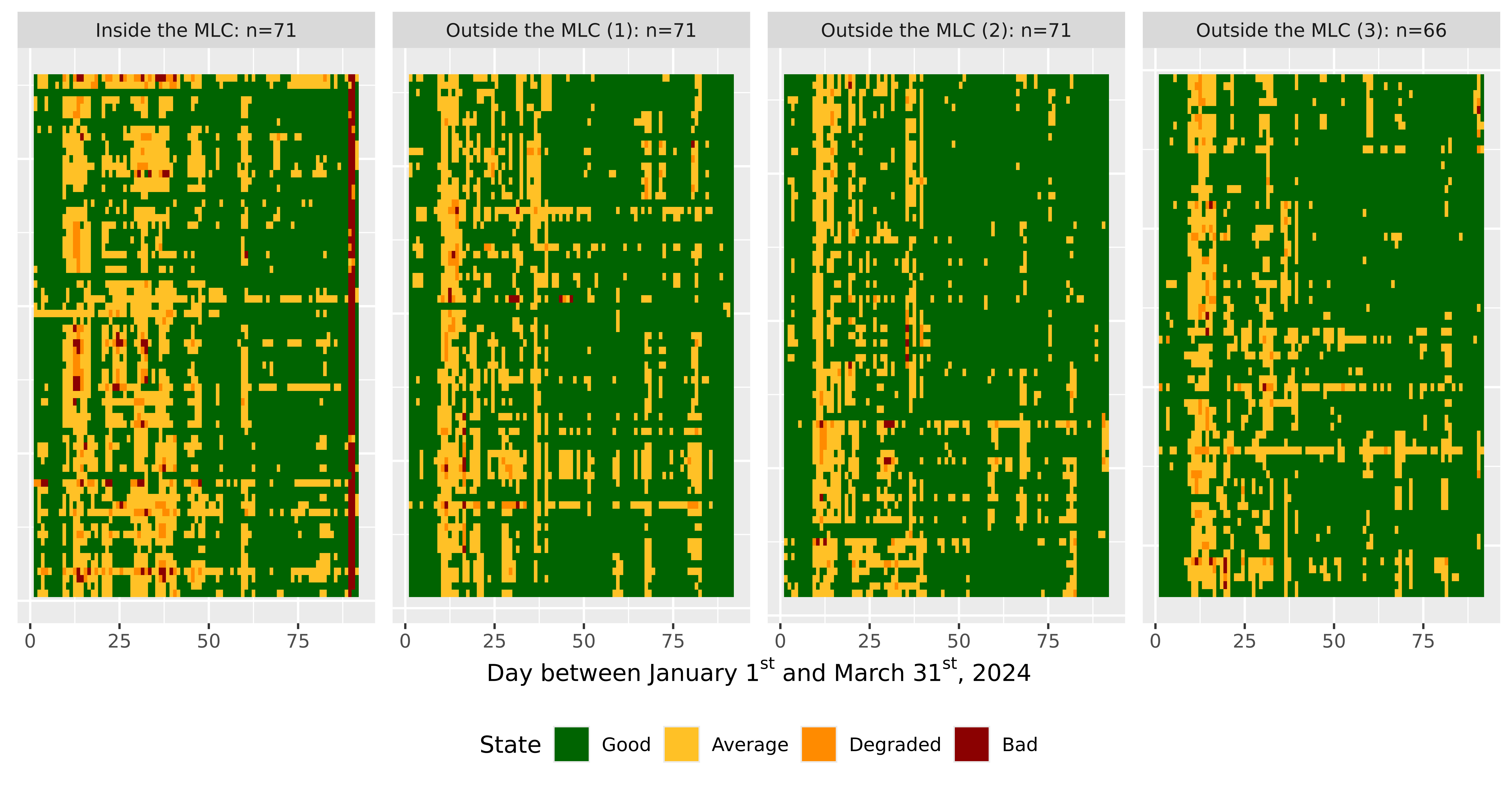}
    \caption{Trajectories inside and outside the MLC detected with the CFSS for the daily mean $\text{PM}_{10}$ concentrations (categorized using the Atmo index).}
    \label{fig:curveMLC}
\end{figure}

\section{Discussion} \label{sec:discussion}

Here, we described our development and testing of a CFSS. This method can identify groups of patients whose trajectories differ significantly from those of others and thus provides a crucial tool for informing public health policymakers.
Previously published methods were not able to detect spatial clusters within categorical, functional data. \\

The results of our simulation study demonstrated that our CFSS has a very satisfactory level of performance. In all of the three scenarios considered, the method correctly identified the simulated spatial cluster with very low false positive rates, high true positive rates, and high positive predictive values. The type I error was maintained at the nominal level. The results of the simulation study also showed that the number of basis functions used when encoding the data has very little impact on performance. This is consistent with the report by \cite{preda2021categorical}. \\

We next applied the CFSS to pollution data; it detected a spatial cluster within which the monitoring stations recorded more transitions to ``average'', ``degraded'' and ``bad'' states than elsewhere and spent more time in these states. \\

It should be noted that we considered circular clusters only. However, the CFSS can be adapted to detect non-circular clusters, such as elliptic clusters \citep{kulldorff2006elliptic} and arbitrary-shaped clusters \citep{tango2005flexibly}. However, these shapes would generate many more potential clusters than the circular approach and would increase the computation time accordingly. \\

Here, we only considered the MLC. However, it would be interesting to investigate secondary clusters, i.e., statistically significant spatial clusters that co-exist with the MLC. An initial approach for detecting secondary clusters was developed by \cite{kulldorff1995spatial}. It consists in estimating the p-value of all potential clusters as if they were the MLC and then considering all statistically significant cluster that do not overlap with a more statistically significant one.
As the concentration indexes of the potential clusters are compared with the maximum of the concentration indexes of each dataset under $\mathcal{H}_0$, the secondary clusters must be strong enough to reject $\mathcal{H}_0$ ``on their own''. \\
Alternatively, \cite{zhang2010spatial}'s sequential approach involves removing the data corresponding to the MLC from the dataset and then iteratively running the scan statistic to find the next most likely cluster. Although other methods have been suggested \citep{li2011spatial, lin2016spatial, xu2016stepwise}, maintenance of the nominal type I error is always a difficulty. \\

The current CFSS can detect clusters based on a single categorical variable. However, in future work, we intend to extend the CFSS to cases involving several categorical variables. \\
The CFSS might also be useful for detecting clusters in datasets composed of qualitative variables and quantitative variables; this will also be addressed in future work.

\appendix

\bibliography{bibliography.bib}
\bibliographystyle{chicago}

\section{Supplementary Materials: Proofs of Section \ref{sec:method}'s results}

\subsection*{Proof of Proposition \ref{prop:propEt}}

\textit{Recall of Proposition \ref{prop:propEt}.} \\
$E_t$ is a self-adjoint and idempotent operator.

\begin{proof} \phantom{newline} \\
Let $Z,Z'$ be two random variables in $L^2$. \\
$\langle Z, E_t(Z') \rangle = \mathbb{E}\left[ Z E_t(Z') \right] = \dsum_{j=1}^J \mathbb{E}\left[ Z \mathbb{E}(Z'|X(t) = e_j) \mathds{1}_{e_j}(t) \right] =  \dsum_{j=1}^J \mathbb{E}(Z'|X(t) = e_j) \mathbb{E}( Z \mathds{1}_{e_j}(t) ) $. \\
Since $\mathbb{E}( Z \mathds{1}_{e_j}(t) ) = \mathbb{E}(Z|X(t) = e_j) \mathbb{P}(X(t)=e_j) $, \\
we deduce $\langle Z, E_t(Z') \rangle = \dsum_{j=1}^J \mathbb{E}(Z'|X(t) = e_j) \mathbb{E}(Z|X(t) = e_j) \mathbb{P}(X(t)=e_j) $. \\
Similarly,
$\langle E_t(Z), Z' \rangle = \dsum_{j=1}^J \mathbb{E}(Z'|X(t) = e_j) \mathbb{E}(Z|X(t) = e_j) \mathbb{P}(X(t)=e_j) $. \\
This shows that $E_t$ is self-adjoint.

Moreover,
\begin{align*}
E_t(E_t(Z)) &= \dsum_{j=1}^J \mathbb{E} \left[ E_t(Z) |X(t) = e_j \right] \mathds{1}_{e_j}(t) \\
&= \dsum_{j=1}^J \mathbb{E} \left[ \dsum_{j'=1}^J \mathbb{E}(Z|X(t) = e_{j'}) \mathds{1}_{e_{j'}}(t) |X(t) = e_j \right] \mathds{1}_{e_j}(t) \\
&= \dsum_{j=1}^J \dsum_{j'=1}^J \mathbb{E}(Z|X(t) = e_{j'}) \mathbb{P} \left[ X(t) = e_{j'} |X(t) = e_j \right] \mathds{1}_{e_j}(t) \\
&= \dsum_{j=1}^J \mathbb{E}(Z|X(t) = e_j) \mathds{1}_{e_j}(t) = E_t(Z)
\end{align*}
This concludes the proof.
\end{proof}

\subsection*{Proof of Proposition \ref{prop:meanzeta}}

\textit{Recall of Proposition \ref{prop:meanzeta}.} \\
$\dint_0^T \mathbb{E}\left[ \zeta(t)^{(1)} \right] \ \text{d}t = \dots = \dint_0^T \mathbb{E}\left[ \zeta(t)^{(M)} \right] \ \text{d}t = 0$.

\begin{proof}
\phantom{newline} \\
Let $m \in \{1, \dots M\}$,
$\zeta^{(m)}$ is orthogonal to $\zeta^{(0)}$. \\
Thus, $\dint_0^T \mathbb{E}\left[ \zeta(t)^{(m)} \right] \ \text{d}t = T \dint_0^T \mathbb{E}\left[ \zeta(t)^{(m)}  \zeta(t)^{(0)} \right] \ \text{d}t = 0$. 
\end{proof}

\subsection*{Proof of Proposition \ref{prop:prop2}}

\textit{Recall of Proposition \ref{prop:prop2}.} \\
$\zeta(t)^{(m)}$ can be written as
$$ \zeta(t)^{(m)} = \dsum_{j=1}^J a_j(t)^{(m)} \mathds{1}_{e_j}(t) $$
where $\{a_j(\cdot)^{(m)}, j = 1, \dots, J \}$ are deterministic encoding functions on $[0,T]$. Thus, finding $\zeta(\cdot)^{(m)}$ that satisfies (\ref{ref:eq2}) is equivalent to solving the more classic problem
\begin{equation}
\dint_0^T \dsum_{j=1}^J \mathbb{P}(X(t') = e_{j'}, X(t) = e_j) a_j(t)^{(m)} \ \text{d}t = \lambda a_{j'}(t')^{(m)} \mathbb{P}(X(t') = e_{j'}). 
\end{equation}
And constraint (\ref{eq:constraint1}) becomes
$$ \dint_0^T \dsum_{j=1}^J \left[ a_j(t)^{(m)} \right]^2 \mathbb{P}(X(t) = e_j) \ \text{d}t = 1. $$

\begin{proof}
\phantom{newline} \\
From $\zeta(t)^{(m)} = \dfrac{1}{\lambda} \dsum_{j=1}^J \mathbb{E}(Z^{(m)}|X(t) = e_j) \mathds{1}_{e_j}(t)$, it follows that $\zeta(t)^{(m)}$ can be written as
$$ \zeta(t)^{(m)} = \dsum_{j=1}^J a_j(t)^{(m)} \mathds{1}_{e_j}(t) $$
with $a_j(t)^{(m)} = \dfrac{1}{\lambda} \mathbb{E}(Z^{(m)}|X(t) = e_j)$. \\
Then Equation (\ref{ref:eq2}) becomes
\begin{align*}
\dint_0^T K(t',t) \dsum_{j=1}^J a_j(t)^{(m)} \mathds{1}_{e_j}(t) \ \text{d}t &= \lambda \dsum_{j'=1}^J a_{j'}(t')^{(m)} \mathds{1}_{e_{j'}}(t') \\
\dint_0^T \dsum_{j=1}^J a_j(t)^{(m)} E_{t'}\left[ E_t(\mathds{1}_{e_j}(t)) \right] \ \text{d}t &= \lambda \dsum_{j'=1}^J a_{j'}(t')^{(m)} \mathds{1}_{e_{j'}}(t') \\
\dint_0^T \dsum_{j=1}^J \dsum_{j'=1}^J a_j(t)^{(m)} \mathbb{P}(X(t) = e_j | X(t') = e_{j'}) \mathds{1}_{e_{j'}}(t') \ \text{d}t &= \lambda \dsum_{j'=1}^J a_{j'}(t')^{(m)} \mathds{1}_{e_{j'}}(t') \stepcounter{equation}\tag{\theequation}\label{eq_aj}
\end{align*}
If $\{a_j(\cdot)^{(m)}\}$ satisfies for all $1 \le j' \le J$:
\begin{equation}
\dint_0^T \dsum_{j=1}^J a_j(t)^{(m)} \mathbb{P}(X(t) = e_j | X(t') = e_{j'}) \ \text{d}t = \lambda a_{j'}(t')^{(m)},
\label{eq:eq3}
\end{equation}
then it satisfies (\ref{eq_aj}).
And from Assumption $(\mathcal{A}_2)$, (\ref{eq:eq3}) is equivalent to
\begin{equation}
\dint_0^T \dsum_{j=1}^J a_j(t)^{(m)} \mathbb{P}(X(t) = e_j, X(t') = e_{j'}) \ \text{d}t = \lambda a_{j'}(t')^{(m)} \mathbb{P}(X(t') = e_{j'}).
\label{eq:eq4}
\end{equation}
Constraint (\ref{eq:constraint1}) can be written as
\begin{align*}
\dint_0^T \mathbb{E}\left[ \dsum_{j=1}^J \dsum_{j'=1}^J a_j(t)^{(m)} a_{j'}(t)^{(m)} \mathds{1}_{e_j}(t) \mathds{1}_{e_{j'}}(t) \right] \ \text{d}t &= 1 \\
\dint_0^T \dsum_{j=1}^J \left[a_j(t)^{(m)}\right]^2 \mathbb{P}(  X(t) = e_j) \ \text{d}t &= 1.
\end{align*}
This concludes the proof.
\end{proof}

\subsection*{Proof of Proposition \ref{prop:prop3}}

\textit{Recall of Proposition \ref{prop:prop3}.} \\
From $\zeta(t)^{(0)}, \dots, \zeta(t)^{(M)}$, the solutions $Z^{(0)}, \dots, Z^{(M)}$ of (\ref{ref:eq1}) can be deduced by
$$ Z^{(m)} = \dint_0^T \zeta(t)^{(m)} \text{d}t. $$

\begin{proof}
\phantom{newline} \\
Let $Z^{(m)} = \dint_0^T \zeta(t)^{(m)} \text{d}t$, then
\begin{align*}
\dint_0^T E_t \left[ Z^{(m)} \right] \text{d}t &= \dint_0^T E_t \left[ \dint_0^T \zeta(t')^{(m)} \text{d}t' \right] \text{d}t \\
&= \dint_0^T \dint_0^T E_t \left[ \zeta(t')^{(m)} \right] \text{d}t' \text{d}t
\end{align*}
Next, from $\zeta(t')^{(m)} = \dsum_{j=1}^J a_j(t')^{(m)} \mathds{1}_{e_j}(t')$,
\begin{align*}
E_{t'}(\zeta(t')^{(m)}) &= \dsum_{j=1}^J a_j(t')^{(m)} \dsum_{j'=1}^J \mathbb{P}(X(t') = e_j | X(t') = e_{j'}) \mathds{1}_{e_{j'}}(t') \\
&= \dsum_{j=1}^J a_j(t')^{(m)} \mathds{1}_{e_j}(t') = \zeta(t')^{(m)}.
\end{align*}
Thus, 
\begin{align*}
\dint_0^T E_t \left[ Z^{(m)} \right] \text{d}t &= \dint_0^T \dint_0^T E_t E_{t'} \left[ \zeta(t')^{(m)} \right] \text{d}t' \text{d}t \\
&= \dint_0^T \lambda^{(m)} \zeta(t)^{(m)} \text{d}t = \lambda^{(m)} Z^{(m)}.
\end{align*}

From Equation (\ref{ref:eq1}), we get that for $\zeta(t)^{(m)} = \dfrac{1}{\lambda} E_t(Z^{(m)})$ solution of (\ref{ref:eq2}),
$$ \dint_0^T E_t(Z^{(m)}) \text{d}t = \lambda^{(m)} Z^{(m)}.$$
Then, $\dint_0^T \lambda^{(m)} \zeta(t)^{(m)} \text{d}t = \lambda^{(m)} Z^{(m)}$, and can be written as
$$ Z^{(m)} = \dint_0^T \zeta(t)^{(m)} \text{d}t . $$
\end{proof}

\subsection*{Proof of Proposition \ref{prop:prop4}}

\textit{Recall of Proposition \ref{prop:prop4}.} \\
The non-trivial solutions of (\ref{ref:eq1}) $Z^{(1)}, \dots, Z^{(M)}$ are centered and uncorrelated random variables such that $\mathbb{V}(Z^{(m)}) = \lambda^{(m)}, 1 \le m \le M$.

\begin{proof}
\phantom{newline} \\
$\mathbb{E}(Z^{(m)}) = \dint_0^T \mathbb{E}\left(\zeta(t)^{(m)} \right) \ \text{d}t = 0$ from Proposition \ref{prop:meanzeta}.
Then, for $1 \le m, m' \le M, m \neq m'$, 
\begin{align*}
Cov(Z^{(m)}, Z^{(m')}) &= \mathbb{E}(Z^{(m)} Z^{(m')}) = \mathbb{E}\left( \dint_0^T \dint_0^T \zeta(t)^{(m)} \zeta(t')^{(m')} \ \text{d}t \ \text{d}t' \right) \\
&= \mathbb{E}\left[ \dint_0^T \dint_0^T E_t\left(\zeta(t)^{(m)}\right) E_{t'}\left(\zeta(t')^{(m')}\right) \ \text{d}t \ \text{d}t' \right] \text{ since } E_t(\zeta(t)^{(m)}) = \zeta(t)^{(m)} \\
&= \dint_0^T \dint_0^T  \mathbb{E}\left[ E_t\left(\zeta(t)^{(m)}\right) E_{t'}\left(\zeta(t')^{(m')}\right) \right] \ \text{d}t \ \text{d}t' \\
&= \dint_0^T \dint_0^T  \mathbb{E}\left\{
\zeta(t)^{(m)}
E_t\left[ E_{t'}\left(\zeta(t')^{(m')}\right)  \right] \right\} \ \text{d}t \ \text{d}t' \text{ since } E_t \text{ is self-adjoint} \\
&= \dint_0^T \mathbb{E}\left\{
\zeta(t)^{(m)}
\dint_0^T E_t\left[ E_{t'}\left(\zeta(t')^{(m')}\right)  \right] \text{d}t' \right\} \ \text{d}t \\
&= \dint_0^T \mathbb{E}\left(
\zeta(t)^{(m)}
\lambda^{(m')} \zeta(t)^{(m')} \right) \ \text{d}t \text{ from Equation (\ref{ref:eq2})}  \\
&= \lambda^{(m')} \dint_0^T \mathbb{E}\left(
\zeta(t)^{(m)}
 \zeta(t)^{(m')} \right) \ \text{d}t = 0 \text{ since } \zeta^{(m)} \text{ and } \zeta^{(m')} \text{ are orthogonal}
\end{align*}

\begin{align*}
\mathbb{V}(Z^{(m)}) &= \mathbb{E} \left[ \left(Z^{(m)}\right)^2 \right] \\
&= \mathbb{E}\left[ \dint_0^T \dint_0^T \zeta(t)^{(m)} \zeta(t')^{(m)} \ \text{d}t \ \text{d}t' \right] \\
&= \mathbb{E}\left[ \dint_0^T \dint_0^T E_t\left(\zeta(t)^{(m)}\right) E_{t'}\left(\zeta(t')^{(m)}\right) \ \text{d}t \ \text{d}t' \right] \text{ since } E_t(\zeta(t)^{(m)}) = \zeta(t)^{(m)} \\
&= \dint_0^T \dint_0^T \mathbb{E}\left[  E_t\left(\zeta(t)^{(m)}\right) E_{t'}\left(\zeta(t')^{(m)}\right) \right] \text{d}t \ \text{d}t'  \\
&= \dint_0^T \mathbb{E}\left\{ \zeta(t)^{(m)}  \dint_0^T E_t\left[ E_{t'}\left(\zeta(t')^{(m)}\right)\right]  \text{d}t' \right\} \text{d}t \text{ since } E_t \text{ is self-adjoint} \\
&= \dint_0^T \mathbb{E}\left\{ \zeta(t)^{(m)}  \lambda^{(m)} \zeta(t)^{(m)} \right\} \text{d}t \text{ from Equation (\ref{ref:eq2})}  \\
&= \lambda^{(m)} \dint_0^T \mathbb{E}\left\{ \left[\zeta(t)^{(m)}\right]^2 \right\} \text{d}t = \lambda^{(m)} \text{ from Constraint (\ref{eq:constraint1})}
\end{align*}
This concludes the proof.
\end{proof}

\section{Supplementary Materials for the simulation study}

\begin{figure}[H]
\centering
\includegraphics[width=0.5\linewidth]{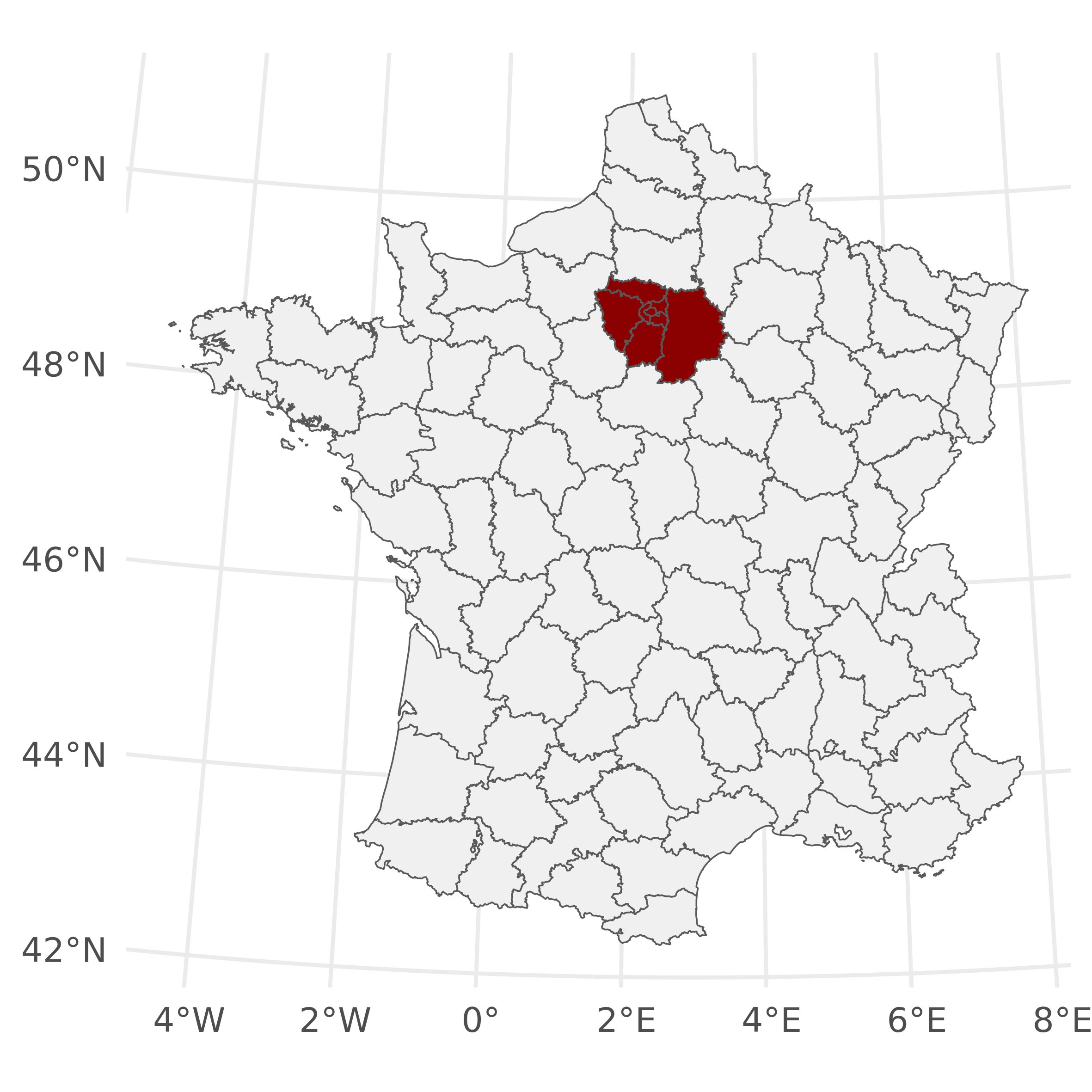}
\caption{The 94 \textit{départements} in mainland France and the spatial cluster (in red) simulated for each artificial dataset.}
\label{fig:map_simu}
\end{figure}

\begin{figure}[H]
\centering
\includegraphics[width=0.8\linewidth]{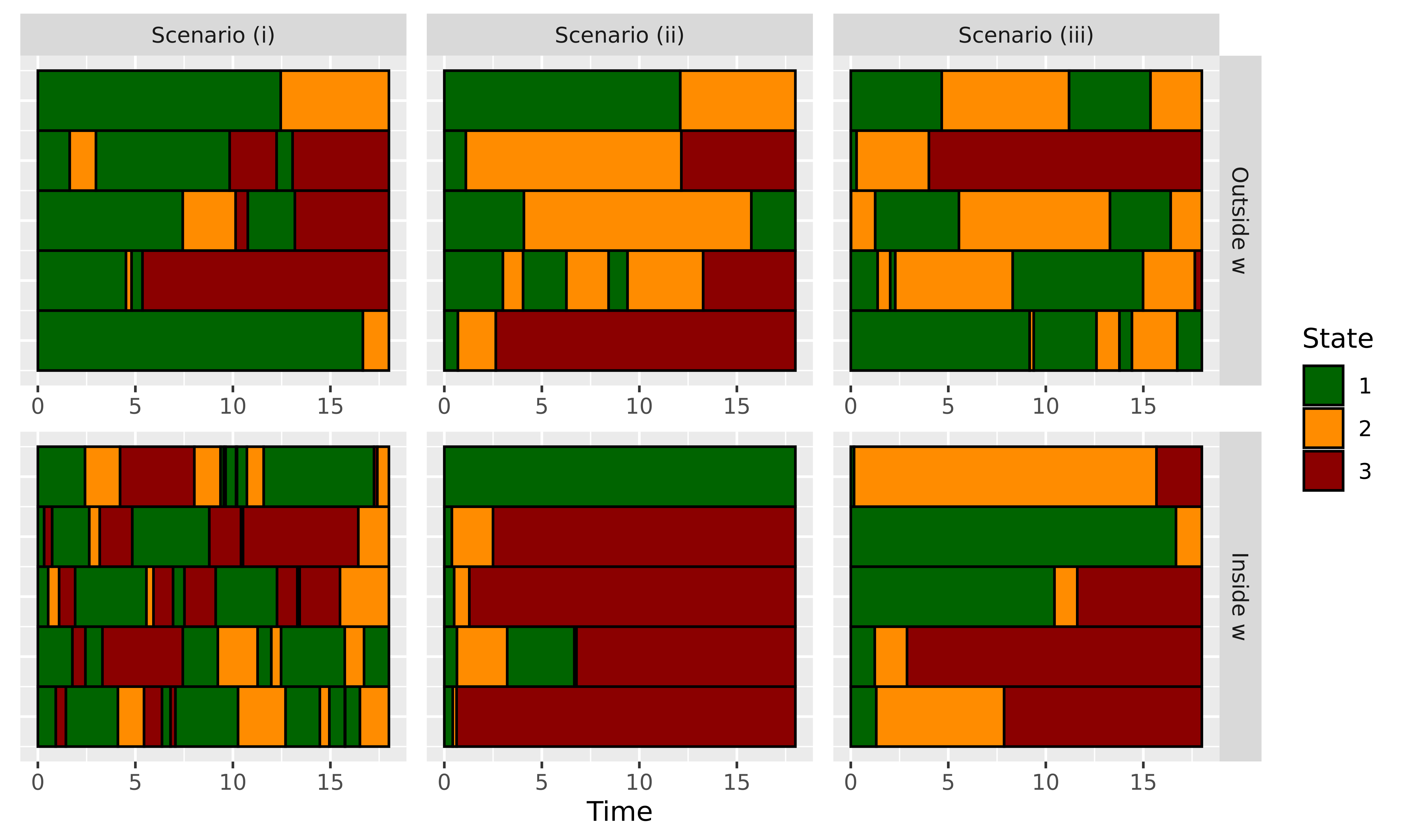}
\caption{Five individual trajectories inside and outside the simulated cluster $w$, for each of three scenarios (scenario (i): $\alpha = 0.6$, scenario (ii): $\alpha = 0.25$, scenario (iii): $\gamma = 0.4$).}
\label{fig:exemple_simu}
\end{figure}

\end{document}